\documentclass{article}
  \let\svthefootnote\thefootnote
\newcommand\blankfootnote[1]{%
  \let\thefootnote\relax\footnotetext{#1}%
  \let\thefootnote\svthefootnote%
}

\usepackage{amsmath, amssymb, amsthm}
\usepackage{graphicx,psfrag,epsf}
\usepackage{enumerate}
\usepackage{placeins}
\usepackage{natbib}
\usepackage{csquotes}

\usepackage{xspace}

\def\spacingset#1{\renewcommand{\baselinestretch}%
{#1}\small\normalsize} \spacingset{1}

\newcommand{\data}{\ensuremath{\textbf{x}}\xspace}
\newcommand{\Data}{\ensuremath{\textbf{X}}\xspace}
\newcommand{\lap}{\ensuremath{\textsf{Lap}}\xspace}

\DeclareMathOperator*{\argmin}{arg\,min}

\newtheorem*{rep@theorem}{\rep@title}
\newcommand{\newreptheorem}[2]{%
\newenvironment{rep#1}[1]{%
 \def\rep@title{#2 \ref{##1}}%
 \begin{rep@theorem}}%
 {\end{rep@theorem}}}
\makeatother

\newreptheorem{theorem}{Theorem}
\newreptheorem{lemma}{Lemma}
\newreptheorem{corollary}{Corollary}

\usepackage{microtype}
\usepackage{graphicx}
\usepackage{subfigure}
\usepackage{booktabs} 
\usepackage[noblocks]{authblk}
\usepackage{xcolor}
\usepackage{hyperref}
\usepackage{algorithmic}
\usepackage{algorithm}
\usepackage[margin=1in]{geometry}

\theoremstyle{plain}
\newtheorem{theorem}{Theorem}[section]

\newtheorem{lemma}[theorem]{Lemma}
\newtheorem{corollary}[theorem]{Corollary}
\theoremstyle{definition}
\newtheorem{definition}[theorem]{Definition}

\theoremstyle{remark}

\usepackage[textsize=tiny]{todonotes}
\title{The Test of Tests: A Framework For Differentially Private Hypothesis Testing}
\author[1]{Zeki Kazan}
\author[2]{Kaiyan Shi}
\author[3]{Adam Groce}
\author[4]{Andrew Bray}
\affil[1]{Department of Statistical Science, Duke University}
\affil[2]{Department of Computer Science, University of Maryland}
\affil[3]{Department of Computer Science, Reed College}
\affil[4]{Department of Statistics, UC Berkeley}
\date{}

\begin{document}
\maketitle

 \vspace{-8mm}
 
\begin{abstract}
We present a generic framework for creating differentially private versions of \textit{any} hypothesis test in a black-box way.  We analyze the resulting tests analytically and experimentally.  Most crucially, we show good practical performance for small data sets, showing that at $\epsilon =1$ we only need 5-6 times as much data as in the fully public setting.  We compare our work to the one existing framework of this type, as well as to several individually-designed private hypothesis tests.  Our framework is higher power than other generic solutions and at least competitive with (and often better than) individually-designed tests.
\end{abstract}

\section{Introduction}

Hypothesis tests are one of the most basic and common statistical analyses that analysts perform on data.  The goal of a hypothesis test is to see whether some ``effect'' in the data (e.g., men are taller than women) is plausibly the result of random variation in the sample, rather than a true fact about the population. Hypothesis tests are the bedrock of statistical analysis in the social sciences, medicine, and other fields, and a variety of hypothesis tests are used, depending on the type of data and the sort of effect one is considering.

However, data in these fields often consists of private information about individuals.  Researchers are under moral and legal obligations to protect the privacy of that data and can often only access that data if they can guarantee their analysis will not violate the privacy of those individuals.  Differential privacy has emerged as the most convincing formal definition of privacy protection in this setting.

Differentially private versions of many popular hypothesis tests have been created, including private analogues of $\chi^2$ tests \citep{fienberg2011privacy, gaboardi2016differentially, johnson2013privacy, rogers2017new, uhlerop2013privacy, vu2009differential, wang2015revisiting}, ANOVA tests \citep{campbell2018differentially, swanberg2019improved}, and many others \citep{alabi2022hypothesis, barrientos2019differentially, canonne2019private, couch2019differentially, ding2018comparing, d2015differential, narayanan2022private, nguyen2017differentially, solea2014differentially, sheffet2017differentially}. However, this is work that privacy researchers must carefully repeat for each possible hypothesis test.  While this is feasible for the most frequently used tests like $\chi^2$ and ANOVA, it is not plausible to expect this work to be repeated for the wide range of hypothesis tests that exist, many of which are highly specific to particular situations.  For example, economists (e.g.~\cite{gatignon1997strategic, jaworski1993market}) use the Chow test to test for a structural break in a regression line;
researchers studying ordinal data (e.g.~\cite{kramer1996ordered, uddin2018factors}) use ordered logistic regression.
Conducting these sorts of analyses privately currently requires collaboration with privacy experts and prohibitive time and effort spent on the study of private statistics before the applied question can even be considered.

In this paper we present a general framework that can automatically create a private version of \textit{any} existing hypothesis test and demonstrate its practicality. For example, at $\epsilon=1$ our method generally requires no more than 5 or 6 times as much data to detect a given effect as would be required in the non-private setting. This makes our off-the-shelf tool competitive with (and occasionally superior to) some individually tailored private hypothesis tests.

\subsection{Our contributions}

The framework we present can be viewed as an instantiation of a method mentioned in the literature, particularly in \cite{canonne2019structure,canonne2019private}. In \cite{canonne2019private}, Cannone et al.~describe the method as follows.
\begin{displayquote}
    ``There exists a black-box method for obtaining a differentially private tester from any non-private tester $A$ using the sample-and-aggregate framework \citep{nissim2007smooth}. Specifically, given any tester $A$ with sample complexity $n$, we can obtain an $\varepsilon$-differentially private tester with sample complexity $O(n/\varepsilon)$."
\end{displayquote}
Unfortunately, these two sentences are the full extent to which this method is considered.  The authors use it only as a point of comparison to show that their method has better asymptotic performance.  We are interested not in the asymptotic case, but in concrete performance that will allow the practical use of private statistics on real data.

The authors quoted above also do not fully specify the method to which they are referring.  Experienced privacy researchers can fill in the details on their own, but they can be filled in in different ways.  Our goal here is to fill in the details completely, giving pseudocode and publicly available implementations, but also to fill in these details in the best possible way and to give concrete analysis of the power of the resulting tests. In particular, we do the following:

\begin{itemize}
    \item We give a framework for creating a private version of any known (non-private) hypothesis test.  This uses the subsample-and-aggregate method, with the aggregation done by the uniformly most powerful binomial test given by \cite{awan2018differentially}.
    \item We give precise analytic expressions for the power of our test in terms of the power of the underlying non-private test.  These finite sample (rather than asymptotic) calculations mean that given a specific public test, one can easily tune parameters in our framework to optimize its power.  These calculations are where we derive the claim that we can get $\epsilon=1$ privacy with 5-6 times the data needed for the public test, but we stress that this is an upper bound \textit{without} test-specific parameter tuning, and in practice our statistical power is often significantly higher.
    \item We implement our framework and use it to privatize several specific tests\footnote{Code that implements our methods is available at \url{https://github.com/diff-priv-ht/test-of-tests}.}.  In particular we consider the context where Cannone et al.~dismissed this method as less powerful than their proposal.  We find that despite their superior asymptotic performance, our framework outperforms their test in a range of practical settings. For example, with a large effect size we obtain 80\% power at $n = 65$, while their test is invalid for $n < 359$ and does not reach 80\% power until $n = 6500$.
\end{itemize}

In concurrent work, Pe{\~n}a and Barrientos \citep{pena2022differentially} also provide a generic framework that implements the idea quoted from Canonne et al.~above.  We delayed the publication of this work to add a full comparison to their framework, which can be found in Section \ref{sec:PB}.  Compared to their framework, ours has meaningfully higher power, and (unlike theirs) can be run for all database sizes, $\epsilon$ values, and choices of public test.

Below, we provide an overview of differentially private hypothesis testing. In section \ref{sec:framework} we outline our test procedure, providing pseudo-code and an analytic expression for the test's power. In Section \ref{sec:PB} we compare our framework to the only existing alternative, that of Pe{\~n}a and Barrientos.  Finally, in Section \ref{sec:tailored} we compare our general framework to some specific existing private tests. 

\section{Background}

In this section, we first discuss hypothesis testing in general. We then introduce differential privacy and the results we will use. 
Finally, we describe prior work on differentially private hypothesis testing.

\subsection{Hypothesis Testing} \label{bk:ht}

Consider a researcher who wants to determine if a new miracle weight-loss drug works as advertised. They measure the weight-loss of individuals in two groups, giving one the drug and one a placebo. They wish to know if the drug had a significant effect. Their first step is to formulate a null hypothesis ($H_0$) - a theory of how the data is distributed. Here $H_0$ may be that the differences in the groups are due to random variation; the drug has no advantage over the placebo.

To test whether or not the data $\data$ is consistent with $H_0$, the researcher will compute a \textit{test statistic} $\tau(\data)$. The choice of a function $\tau$ to compute the test statistic largely determines which hypothesis test being used. For a random database $\Data$ drawn according to $H_0$, the distribution of the statistic $T = \tau(\Data)$ can be determined either analytically or through simulation. The researcher then computes a \textit{p-value}, the probability that the observed test statistic or a more extreme value would occur under $H_0$.

\begin{definition}
For an observed test statistic $t=\tau(\data)$ and null hypothesis $H_0$, the one-sided \emph{p-value}, $p$, is defined as
$$p = \Pr[T \geq t \mid T = \tau(\Data) \text{ and } \Data \leftarrow H_0].$$
\end{definition}

If the function $\tau$ is well-chosen, then the more the underlying distribution of $\Data$ differs from the distribution under $H_0$, the more likely a low p-value will be. Typically a significance threshold $\alpha$ is chosen, and $H_0$ is rejected as a plausible explanation of the data when $p < \alpha$. The choice of $\alpha$ determines the \textit{type I error rate}, the probability of incorrectly rejecting a true null hypothesis. 

We define the \textit{critical value} $t^*$ to be the value of the test statistic $t$ when $p = \alpha$. We use this to define the \textit{statistical power}, a measure of how likely a hypothesis test is to pick up an effect (i.e. to reject a false null hypothesis). The power is a function of how much the underlying distribution of $\Data$ differs from the distribution under $H_0$ as well as the size of the database.

\begin{definition}
For a given alternate data distribution $H_A$, the \emph{statistical power}, $\theta$,  of a hypothesis test is 
$$\theta = \Pr[T \geq t^* \mid T = \tau(\Data) \text{ and } \Data \leftarrow H_A].$$
\end{definition}

The goal of hypothesis test design is to maximize statistical power, ideally finding a single test that has good performance for a range of effects.

\subsection{Differential Privacy}

To convince the public to allow their confidential data to be used for statistical analyses, researchers need to guarantee that sensitive information will not be compromised.  Previous methods adopted to protect individual privacy, such as anonymization, have been shown to fail in numerous cases (e.g. \cite{sweeney2002k, narayanan2008robust, homer2008resolving}).

Differential privacy, proposed in 2006 by Dwork et al. \citep{dwork2006calibrating}, is a formal definition of privacy. It protects an individual's privacy by requiring that any output occurs with roughly equal probability regardless of value of that individual's information. Databases that differ only in the data of one individual are called \textit{neighboring} databases. 

\begin{definition}[Differential Privacy] \label{def:dp}
A randomized algorithm $\tilde{f}$ on databases is $(\varepsilon,\delta)$ differentially private if for all $\mathcal{S} \subseteq \textup{Range}(\tilde{f})$ and
for databases $\data, \data'$ that only differ only in the values of one row:
$$\Pr[\tilde{f}(\data) \in \mathcal{S}] \leq e^\varepsilon \Pr[\tilde{f}(\data') \in \mathcal{S}]+\delta.$$
\end{definition}

It is possible that $\delta = 0$. Under this condition, the randomized algorithm $\tilde{f}$ is said to be $\varepsilon$-differentially private. In general, $\varepsilon$ indicates the privacy level (a smaller $\varepsilon$ indicates a higher privacy guarantee) and $\delta$ determines the likelihood of privacy failure.  An $(\varepsilon, \delta)$-differential privacy guarantees that, with $1-\delta$ probability, the privacy loss is bounded by $e^\varepsilon$. 

Differential privacy is resistant to post processing --- if an algorithm is differentially private, any further analysis or computation on the output (without dependence on the database) will also result in private output.

\begin{theorem}[Post Processing] \label{thm:pp}
 Let $\tilde{f}$ be an $(\varepsilon,\delta)$-differentially private randomized
algorithm. Let $g$ be an
arbitrary randomized algorithm. Then $g \circ \tilde{f}$ is $(\varepsilon,\delta)$-
differentially private.
\end{theorem}

Any differentially private algorithm must be randomized.  The most popular (and simple) method is the \textit{Laplace mechanism}, introduced by Dwork et al. \citep{dwork2006calibrating}, which adds noise drawn from the Laplace distribution to the output of the query one seeks to privatize. 
 
\begin{definition}[Laplace Distribution]\label{def:ld}
The Laplace Distribution
centered at 0 with scale $b$ has probability density function
$$ \lap(x|b)=\frac{1}{2b}\textup{exp}\Big(-\frac{|x|}{b}\Big).$$ We write $\lap(b)$ to denote the Laplace distribution with scale $b$.
\end{definition} 

The magnitude through which the alteration of a single row in the database can change the output of a query is called the global sensitivity.

\begin{definition}[Global sensitivity] \label{def:gs}
The global sensitivity of a function $f$  is:
$$GS_f = \underset{\data,\data'}{ \textup{max}}\ |f(\data)-f(\data')|,$$ where $\data$ and $\data'$ are neighboring databases. 
\end{definition}

The standard deviation of the Laplace Distribution used to introduce noise depends on both $\varepsilon$ and $GS_f$.

\begin{definition}[Laplace Mechanism]  \label{def:lm}
Given any function $f$, the Laplace mechanism is defined as
$$\tilde{f}(\data)= f(\data)+Y,$$ where $Y$ is drawn from $\lap(GS_f/\varepsilon)$, and $GS_f$ is the global sensitivity of $f$.
\end{definition}

\begin{theorem}[Laplace Mechanism] \label{thm:lm}
The Laplace mechanism $(\varepsilon,0)$-differentially private.\end{theorem}

Although the Laplace mechanism ensures that an output will not violate privacy, sometimes the global sensitivity is so large that the Laplace noise overwhelms the signal. The \textit{subsample and aggregate} technique \citep{nissim2007smooth} is designed to mitigate this problem. Subsample and aggregate works exactly as it sounds. The database \data with $n$ rows is first partitioned into $m$ groups of approximately equal size. Then a non-private function $f$ is computed in each group independently. Finally, these intermediate results are aggregated through some differentially private mechanism.

\subsection{Related Works}

There is an extensive (and rapidly expanding) literature
examining the problem of converting public hypothesis tests
to the private setting. One line of work \citep{smith2008efficient, smith2011privacy, wasserman2010statistical} studies how fast the distributions of private test statistics converge to the public. These results, however, are often asymptotic and offer little in the way of implementable tests. 
\cite{wang2018statistical} studies the problem of generating a reference distribution more thoroughly, providing a general recipe for approximating the sampling distributions of private test statistics.

Another line of work examines the problem of privatizing the test statistic for the $\chi^2$ test of independence. This includes works in the context of genome-wide association study (GWAS) data \citep{fienberg2011privacy, johnson2013privacy, uhlerop2013privacy}, although they tend to use asymptotic arguments for the uniformity of p-values. Other work \citep{gaboardi2016differentially, wang2015revisiting} has shown that Monte Carlo methods can produce better reference distributions. \cite{vu2009differential} 
provides concrete methods for producing a p-value by adjusting for Laplace noise, while 
\cite{rogers2017new} proposes alternate test statistics that have reference distributions with preferable properties.

Recent works in differentially private hypothesis testing have begun to include in-depth power analyses.  
\cite{awan2018differentially} constructed the universally most powerful test for binomial data (see Section \ref{fwk:binom} for further discussion). 
\cite{brenner2010impossibility} shows that a universally most powerful test cannot exist for data with a domain containing more than two elements.  Nguy{\^e}n and Hui 
 propose methods for differentially private survival analysis \citep{nguyen2017differentially}, two works have addressed the problem of studying the difference in means of normal distributions \citep{ding2018comparing, d2015differential}, and several consider the problem of hypothesis testing for linear regression coefficients \citep{alabi2022hypothesis, barrientos2019differentially, sheffet2017differentially}.  Of these, \cite{barrientos2019differentially} is notable for sharing some conceptual ideas with the framework we propose here. A few works propose tests for the mean of a normal distribution in the univariate \citep{solea2014differentially} and mulivariate \citep{canonne2019private, narayanan2022private} settings. Two works study the one-way ANOVA \citep{campbell2018differentially, swanberg2019improved}, although these are outperformed by work on nonparametric alternatives \citep{couch2019differentially}. 
\cite{avella2021privacy} proposes a hypothesis test based on $M$-estimators that is applicable to general parametric models, such as many of the above.

\section{Framework} \label{sec:framework}

Here we introduce our test of tests (ToT) framework and analyze its power.  We also discuss how to optimize the framework's parameters for a given situation.

\subsection{Private Binomial Test} \label{fwk:binom}

Awan and Slavkovi\'c \citep{awan2018differentially} develop a uniformly most powerful test for binomial data. They define the Truncated-Uniform-Laplace (Tulap) Distribution, 
the sum of the discrete Laplace and uniform distributions.
The distribution is parameterized by a location parameter, $m$, and a scale parameter, $b \in (0,1)$. Its CDF has a closed form; see Definition 4.1 in \cite{awan2018differentially}.
\footnote{The Tulap distribution has a third parameter, $q$, but we always set $q=0$ because our aim is to have $\delta=0$.  Allowing $\delta > 0$ could be done by changing $q$, and would increase the power of our test.}

Let $A \sim \textup{Binomial}(n,p)$. Awan and Slavkovi\'c show that the private test statistic $Z|A \sim \textup{Tulap}(A, e^{-\varepsilon})$ is an $\varepsilon$-differentially private estimate of $A$. 
They also provide an algorithm for producing a p-value to test the hypothesis
\begin{align*}
    H_0: p \leq p_0 \quad \textup{and} 
\quad H_A: p > p_0
\end{align*}
and show that the p-value produced is the smallest $\varepsilon$-DP p-value for this test. See Theorem 7.2 and Algorithm 2 in \cite{awan2018differentially} for further details.

\subsection{Our Algorithm}

\begin{figure}[t]
    \centering
    \includegraphics[width=0.85\columnwidth, keepaspectratio]{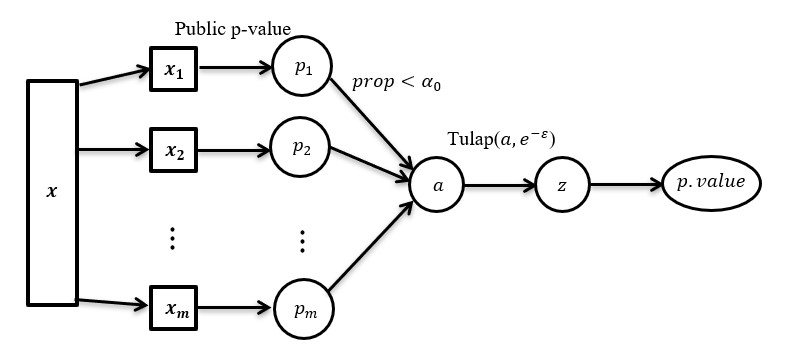}
    \caption{A graphical representation of Algorithm \ref{alg:general}.}\label{fig:tot_diagram}
\end{figure}

We now describe our general algorithm, which we call \textit{test of tests} (ToT), which can privatize all hypothesis tests. The formalization is presented in Algorithm \ref{alg:general} and a graphical representation in Figure \ref{fig:tot_diagram}. We are given a database $\data$ of size $n$, and our goal is to run an $\varepsilon$-private version of hypothesis test $\tau$ \footnote{For concision, we use $\tau$ to represent a test that utilizes $\tau(\data)$.} on that database with significance threshold $\alpha$.  We first partition the input database into $m$ equal sized subsets $\data_1, \ldots, \data_m$. In practice, if $m \nmid n$, then the subsets should be of sizes $\lfloor \frac{n}{m} \rfloor$ and $\lceil \frac{n}{m} \rceil$ as appropriate. The following results will assume that $m \mid n$ for simplicity. In each subset, we conduct the public test $\tau$, computing the p-value and accepting/rejecting according to a sub-test significance threshold of $\alpha_0$. If the number of data points in a subsample is insufficient to run the public test, the p-value is drawn from $\textup{Unif}(0,1)$. Let $a$ be the number of rejects.  Under the null distribution, each instance of $\tau$ rejects with probability $\alpha_0$, so $a$ follows a binomial distribution.

We then conduct Awan and Slavkovi\'c's private binomial test on $a$ to see if it is consistent with a binomial distribution with parameter $\alpha_0$. To privatize $a$, we define $z = \textup{Tulap}(a, e^{-\varepsilon})$. Let $B \sim \textup{Binomial}(m, \alpha_0)$ and $N \sim \textup{Tulap}(0, e^{-\varepsilon})$. Then, the reference distribution is $B+N$ and so the p-value is $P(B + N \geq z)$.

\begin{algorithm}[t]
\caption{Test of Tests}
\label{alg:general}
\begin{algorithmic}
    \STATE {\bfseries Input:} $\data$, $\tau$, $\varepsilon$, $\alpha$, $m$, $\alpha_0$
    \STATE  Partition $\data$ into subsets $\data_1, \ldots, \data_m$ \;
    \FOR{$j = 1$ to $m$}
        \IF{$\tau$ can be run on $\data_j$}
        \STATE $p_j \longleftarrow \tau(\data_j)$ 
        \ELSE
        \STATE $p_j \sim \textup{Unif}(0,1)$
        \ENDIF
    \ENDFOR
    \STATE $a \longleftarrow$ $|\{ p_j : p_j < \alpha_0 \}|$
    \STATE $z \longleftarrow \textup{Tulap}(a, e^{-\varepsilon})$
    \STATE $p.value \longleftarrow P(B + N \geq z)$
    \STATE {\bfseries Output:} $z$, $p.value$
\end{algorithmic}
\end{algorithm}

Note that of the inputs listed, $\data$, $\tau$, $\epsilon$, and $\alpha$ are true inputs from the user, while $m$ and $\alpha_0$ are parameters that can be optimized.  We discuss this optimization in Section \ref{fwk:optim}.

The privacy and validity of Algorithm \ref{alg:general} follow immediately from its design.  

\begin{theorem} \label{thm:priv}
    Algorithm \ref{alg:general} is $\varepsilon$-differentially private.
\end{theorem}

\begin{proof}
    By Subsample and Aggregate \citep{nissim2007smooth} and Theorem 6.1 in Awan and Slavkovi\'c \citep{awan2018differentially}, which shows the release of the statistic with Tulap noise satisfies privacy, the release of $z$ is $\varepsilon$-differentially private. By Theorem \ref{thm:pp} (post processing), the release of the p-value is also $\varepsilon$-differentially private.
\end{proof}

\begin{theorem} \label{thm:valid}
    Algorithm \ref{alg:general} is valid.  That is, when the data is drawn from $H_0$ the probability of rejection is at most $\alpha$.
\end{theorem}

\begin{proof}
    Each of the $m$ subgroups will reject (i.e., be included in the count $a$) with probability at most $\alpha$.  For most this follows from the validity of the public test $\tau$.  In cases when $\tau$ can't be run, it follows from the uniform selection of $p_j$.  From there, the validity follows immediately from the results of \cite{awan2018differentially}.
\end{proof}

\subsection{Theoretical Power}

We can now analyze the statistical power of the test of tests framework.  We begin by noting its asymptotic sample complexity as a function of $\varepsilon$.  This was stated without proof by \cite{canonne2019private}, and we provide a proof in Appendix \ref{sec:proofs}.

\begin{theorem} \label{thm:sc}
The number of samples required for our test to achieve $\rho$ power is $n = \mathcal{O} \left(c/\varepsilon\right)$, where $c$ is the number of samples needed by the non-private test, $\tau$.
\end{theorem}

The focus of this work is not asymptotic performance, but practical performance on small $n$, and for that analysis we need an exact computation of the power of any ToT instantiation.

\begin{theorem} \label{thm:power}
Let $\theta$ be the power of the public test $\tau$ in each of the $m$ subsamples with significance level $\alpha_0$. Let $A \sim \textup{Binomial}(m,\theta)$, $Z|A \sim \textup{Tulap}(A, e^{-\varepsilon})$, $B \sim \textup{Binomial}(m, \alpha_0)$, and $N \sim \textup{Tulap}(0, e^{-\varepsilon})$.
Then the power of our test is
$$\mathcal{P}(\varepsilon, \alpha, m, \alpha_0, \theta) = (1 - F_{Z}(F^{-1}_{B+N}(1-\alpha))).$$
\end{theorem}

The proof of this theorem is messy, so for clarity we consign it to Appendix \ref{sec:proofs}.  Note that $F_{B + N}^{-1}$ does not have a known analytic form, so when computing the power via Theorem \ref{thm:power}, the quantiles of the distribution must be determined numerically.

If one is interested in a particular public hypothesis test with known characteristics, the above result can be used to determine a bound on the sample size required for the privatized test to achieve $\rho$ power.  (Simple proof in Appendix \ref{sec:proofs}.)

\begin{corollary} \label{cor:general_power}
Suppose that a public hypothesis test $\tau$ requires at most $n$ data points to achieve $\theta$ power at a significance level $\alpha_0$ for any choice of the data. Then, in order for the private test with privacy parameter $\varepsilon$ to achieve $\rho$ power at a significance level $\alpha$, the necessary number of data points is bounded above by $n\tilde{m}$, where $\tilde{m}$ is the smallest $m$ such that $ \rho \leq \mathcal{P}(\varepsilon, \alpha, m, \alpha_0, \theta)$.
\end{corollary}

Since the power $\mathcal{P}$ is strictly increasing with respect to $m$, it is straightforward to determine $\tilde{m}$ numerically. This allows general statements about how much more data a private test will need compared to the equivalent public test.  Some examples are shown in Table 1.  For example, the first row shows that \textit{any} public test that achieves 80\% power at $\alpha=0.05$ can be privatized at $\varepsilon=1$ (by using exactly that public test as the subtest) to get the same power and significance with $\tilde{m} = 5$, meaning that the private test needs 5 times the data of the public test.  For 95\% power 6 times the data of the public test is needed.  (For $\varepsilon = 0.01$ those multiples are 44x and 52x respectively.)

We stress that these general statements, while they are very strong, are only upper bounds.  That is because without specifying a test, one cannot say what would happen when the $\alpha_0$ for the subtests is different than the $\alpha$ one is attempting to achieve in the overall test.  Given any particular test, one can vary $\alpha_0$ and find better settings.  For example, a z-test with $\alpha=0.05$ run on data with an effect size of 0.65 standard deviations will reach 80\% power at $n=20$, meaning that the statement above would guarantee no more than $n=100$ needed to get the same power in the $\varepsilon=1$ private setting.  But allowing $\alpha_0$ to take values other than 0.05, we find that one can actually do this with $n=70$, meaning a $3.5\times$ cost of privacy, rather than $5\times$.  At $\varepsilon=0.01$, it requires $n=420$, for a $21 \times$ cost of privacy, instead of the 44 given by the upper bound in the table.

As another example, take an ANOVA test with three groups run on data with equal within-group and between-group variance.  The upper bounds in the table for 95\% power require $6\times$ data at $\varepsilon=1$ and $52\times$ data at $\varepsilon=0.01$, but the optimized test requires $3.6\times$ and $41\times$ data instead.

\begin{table}
\centering
\begin{tabular}{ll|lll|l}
$\theta$ & $\alpha_0$  & $\rho$ & $\alpha$ & $\varepsilon$ & $\tilde{m}$  \\
\hline
0.80     & 0.05         & 0.80      & 0.05     & 1             & 5     \\
0.80     & 0.05         & 0.80      & 0.05     & 0.1           & 44     \\
0.95     & 0.05        & 0.95      & 0.05     & 1             & 6  \\
0.95     & 0.05        & 0.95      & 0.05     & 0.1             & 52 
\end{tabular}
\caption{For a public test $\tau$ that requires achieves $\theta$ power at significance level $\alpha_0$, in order for our $\varepsilon$-private test to achieve $\rho$ power at a significance level of $\alpha$, we require at most a factor of $\tilde{m}$ more data.}
\label{tbl:pow_thm}
\end{table}

\subsection{Optimization} \label{fwk:optim} 

The variables $m$ (the number of subsamples) and $\alpha_0$ (the sub-test significance threshold) must be optimized. Fortunately, we find that doing an extremely thorough optimization for these parameters is not necessary. The optimal $m,\alpha_0$ combination for one effect size generally does an adequate job across a large range of effect sizes, with a decrease in power generally in the range of 1 to 2\%. 

When an approximate expected effect size is known, we can easily compute $\theta$, the power of the public test $\tau$, for any sample size. For fixed $m$, standard techniques can be used to find the $\alpha_0$ that maximizes $\mathcal{P}(\varepsilon, \alpha, m, \alpha_0, \theta)$.  This can then be repeated for all $m$ in a reasonable set to find the otpimal $m,\alpha_0$ pair. For our simulations, we use the set $\{1, 2, 3, \ldots, \lfloor{\sqrt{n}}\rfloor, \ldots \lfloor{\frac{n}{3}}\rfloor, \lfloor{\frac{n}{2}}\rfloor, n\}$ and find that this process takes less than 20 seconds for a t-test at $n = 100$.

In practice, however, an approximate expected effect size is often not known a priori. In this setting, we suggest fixing a desired power, $\rho$, and optimizing for the $m,\alpha_0$ pair that minimizes the effect size detectable with $\rho$ power. This can be achieved by beginning with a grid of effect sizes and performing a binary search, using the above process for known effect size at each step, to find the minimum effect size in the grid detectable with $\rho$ power. Then use the $m,\alpha_0$ from that combination to run the test of tests. In our simulations, we use a length-$16$ grid and find this process takes less than a minute and a half for a t-test at $n = 100$.

We note, interestingly, that the optimization tends to favor high values of $m$, with very small subsamples and high significance thresholds on the subtests.  It turns out that aggregating a large number of minimally-informative tests is preferable to a small number of more reliable tests.

\section{Comparison to  Pe{\~n}a-Barrientos Framework} \label{sec:PB}

Pe{\~n}a and Barrientos \citep{pena2022differentially}, simultaneously to this work, proposed their own framework (henceforth referred to as PB) for privatizing arbitrary public hypothesis tests.  Like our framework, theirs follows the subsample-and-aggregate idea first mentioned by Canonne et al.~\citep{canonne2019structure}.  

Both methods begin by running the public test on subsamples of the data set, but the methods of aggregation are different.   PB develop what is essentially a custom-built binomial test based on a randomized response-type method.  We instead use the Awan and Slavkovi\'{c} binomial test, which is provably optimal.  As a result, our framework is the highest-power framework possible within this general type of design.

\begin{theorem} \label{thm:PB_vs_ToT}  For any choice of public test $\tau$ and privacy parameter $\epsilon$, the statistical power of the private test resulting from the ToT framework will be higher than that resulting from the PB framework.
\end{theorem}

\begin{proof} Fix a number of subtests $m$ and subtest significance threshold $\alpha_0$. Then the higher power for test of tests is an immediate consequence of the main result of Awan and Slavkovi\'{c} \citep{awan2018differentially}.  Up until the end of the for-loop in Algorithm \ref{alg:general}, the two frameworks are identical, and the remainder of the algorithm can be viewed as a binomial test for whether the proportion of ``reject'' decisions in subtests is greater than $\alpha_0$.  Because the  Awan and Slavkovi\'{c} binomial test is proven to be the uniformly most powerful test in this situation, it must be higher power than the actions performed by the PB framework.

Allowing test of tests to use it's own optimal $m$ and $\alpha_0$ (rather than matching that of PB) can only increase the gap between the two frameworks, since the optimal values might differ.
\end{proof}

The PB framework has drawbacks beyond the simple lowering of power.  The details of the test mean that it cannot get valid results at all parameter settings.  In particular, there is a minimum $m$ value at which the test can be run.  Since the public test itself often requires a certain amount of data, this means that a meaningful amount of data is sometimes required before the PB test can be used at all.  (For example, with $\varepsilon=1$ and $\alpha=0.05$, PB requires $m \geq 7$, increasing to $m\geq 67$ when $\varepsilon=0.1$.)

\begin{figure}[t]
    \centering
    \includegraphics[width=\columnwidth]{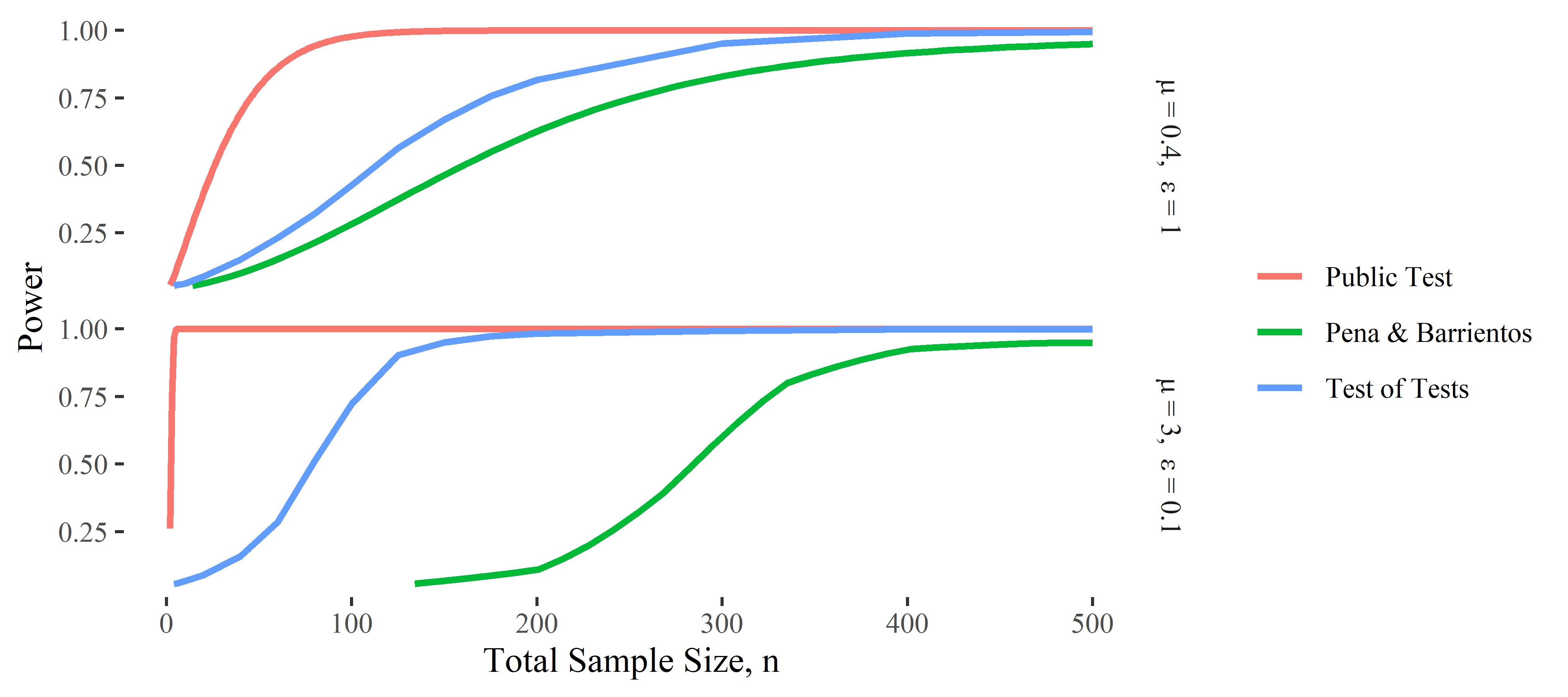}
    \caption{Power comparison between \cite{pena2022differentially} and the test of tests for various sample sizes $n$ with $\alpha = 0.05$ and a t-test with $\sigma = 1$. The top panel has an effect size of $\mu = 0.4$ and $\varepsilon = 1$; the bottom has $\mu = 3$ and $\varepsilon = 0.1$. We optimize $m,\alpha_0$ for the test of tests as discussed in Section \ref{fwk:optim} with target $\rho = 0.9$.}
    \label{fig:PB_vs_ToT}
\end{figure}

Furthermore, given values of the other parameters, $\alpha_0$ \textit{must} be set to a specific value so that the resulting $\alpha$ of the larger test is accurate.  This removes a degree of freedom in optimization, further worsening power.  PB give two methods of setting parameters.  The first involves no real optimizing at all, suggesting that $m$ be set as low as possible.  This is meant for the ``low-power'' setting, where the goal is to achieve significant power at the lowest possible $n$.  (In contrast, we find that very high $m$ often results in better power.)  We use this method in our comparison calculations, though we note that they also suggest one could calculate power curves at a variety of parameter values and choose the best 
parameters through visual inspection.  This method is necessary to reach high power, because using the lowest possible $m$ results in an upper bound on the power of the PB test, meaning that the power does not approach 100\% as $n$ grows.  (This bound can be as low as 80\% in realistic scenarios.)

In Appendix \ref{sec:proofs} we provide an analogue of our Theorem \ref{thm:power} for the PB test so that we can directly compare power instead of relying upon approximate simulations. Figure \ref{fig:PB_vs_ToT} shows the exact power of PB and of our ToT framework for two examples.  We use a t-test as the public test.  In the top panel, with a moderate effect size and $\varepsilon=1$, the PB framework requires 40\% more data to achieve 80\% power ($n=200$ for ToT, $n=280$ for PB).  In the bottom panel, with a larger effect size and $\varepsilon=0.1$, the difference is much greater.  Here ToT only requires $n=125$ to get 80\% power, while PB cannot be run at all until $n \geq 134$ and doesn't get 80\% power until $n=348$ (a 178\% increase).  Additional comparisons privatizing a z-test and an ANOVA can be found in Appendix \ref{sec:add_figs}.

\section{Comparisons to Tailored Tests}\label{sec:tailored}

In this section, we demonstrate the use of the ToT framework on a selection of hypothesis tests, namely a test for the mean of multivariate normal data and a one-way ANOVA. These tests have both been the subject of prior work, so we can compare our general-purpose technique to tests experts carefully developed for specific situations.

\subsection{Mean of Multivariate Normal Data} \label{ex:norm}

Since the general method we are using was first mentioned (and dismissed) in \cite{canonne2019structure,canonne2019private}, we begin by using our framework to develop a test for the same situation.  Here the analyst observes data drawn from a multivariate normal distribution, $\data = \{ \mathbf{X}_1, \ldots, \mathbf{X}_n \}$ and $\mathbf{X}_i \sim \mathcal{N}_d(\boldsymbol{\mu}, \mathbb{I}_d)$.  The null hypothesis is that $\boldsymbol{\mu} =\mathbf{0}$, while the alternate has $\boldsymbol{\mu} \neq \mathbf{0}$.
A public hypothesis test for this setting uses the test statistic $Z = n\sum_{j = 1}^d \bar{X}_j^2$,
which is known to follow the distribution $\chi^2(df = d)$. 
We will use this test for the p-value computation step in Algorithm \ref{alg:general}. We also compute the power of this test.  (The proof of the theorem below is in Appendix \ref{sec:proofs}.)

\begin{theorem} \label{thm:pow_norm}
Let $F_0$ and $F_A$ be the CDFs of $\chi^2(df = d)$ and $\chi^2 \left(df = d, \lambda = n\sum_{j = 1}^d \mu_j^2 \right)$, respectively, where $n$ is the sample size and $\mu_j$ is the $j^{th}$ entry of $\boldsymbol{\mu}$. Then the power of the public test with significance level $\alpha$ is
$1 - F_{A}(F_0^{-1}(1-\alpha))$.
\end{theorem}

This analytic expression for the power is not needed to perform the test, but having it allows optimization to be done more efficiently and means our figures show the exact power of our test, rather than a Monte Carlo approximation.

\cite{canonne2019private} proposes a computationally efficient private test for this setting\footnote{\cite{narayanan2022private} points out an error in this work, but it is in a second algorithm that is irrelevant to this comparison.} and proves that its asymptotic dependency on $\varepsilon$ and effect size is superior to the strategy we use here.  \cite{narayanan2022private} gives another test with yet better asymptotic performance.  Unfortunately, this algorithm is described only in general asymptotic terms, without the concrete details necessary for implementation.  As a result, we compare to the test given by Canonne et al.

\begin{figure}[t]
    \centering
    \includegraphics[width=\columnwidth]{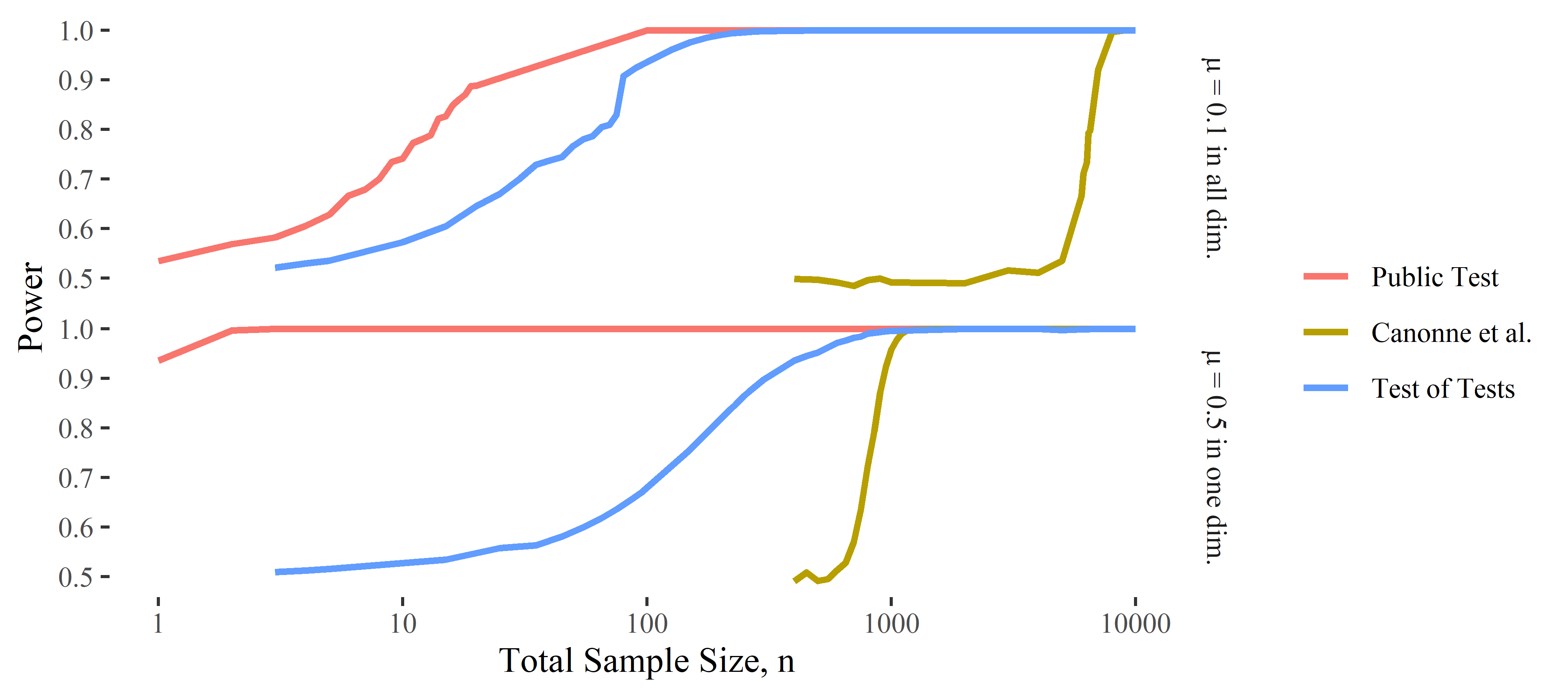}
    \caption{Power comparison between \cite{canonne2019private} and the tests of tests for various sample sizes $n$. For the top panel, $\mu_i = 0.1$ for all $i$; for the bottom, $\mu_1 = 0.5$ and $\mu_i = 0$ for $i \neq 1$. We set $d=100$, $\varepsilon = 1$, and $\alpha$ for the test of tests and public test is set to match the Type I Error of the \cite{canonne2019private} test. We optimize $m,\alpha_0$ for the test of tests as discussed in Section \ref{fwk:optim} with target $\rho = 0.9$ for $n \leq 500$ and $\rho=0.99$ for $n > 500$.}
    \label{fig:Canonne}
\end{figure}

The Canonne test does not have an adjustable $\alpha$ value.  Instead, the analyst is inputs a parameter, $\gamma$,\footnote{We call this parameter $\gamma$, rather than the $\alpha$ from the paper to avoid conflict in notation.} that is a lower bound on the total variation distance between the null distribution $\mathcal{N}_d(\boldsymbol{0}, \mathbb{I}_d)$ and the alternate distribution $\mathcal{N}_d(\boldsymbol{\mu}, \mathbb{I}_d)$.  The test is then guaranteed to distinguish the two distributions with probability 2/3 with a required sample complexity of 
$ \tilde{\mathcal{O}}\left(d^{1/2}/\gamma^2 + d^{1/2}/(\gamma\varepsilon)\right).$ 
This means that Type 1 error will approach 1/3 for sufficiently high $n$, but it can be much higher at low $n$.  In fact, there is a threshold of $\max \left\{25 \log \frac{d}{\delta}, \frac{5}{\varepsilon} \log\frac{1}{\delta}\right\}$ below which the test always rejects (100\% Type 1 error).  Just above this threshold it has Type 1 error of roughly 50\%, where it remains for the entire range of $n$ values we are considering.  (This Type 1 error is proven analytically in Appendix \ref{sec:proofs} and confirmed experimentally in Appendix \ref{sec:add_figs}.)  As a result, we set the $\alpha$ value in our test to 0.5 for a fair comparison, but we note that our test has the advantage that $\alpha$ can be set arbitrarily.

For our comparison, we set $d=100$ and $\varepsilon=1$.  ToT uses pure differential privacy, with $\delta=0$, but Canonne requires a nonzero $\delta$.  We set $\delta=10^{-3}$, which we believe to be very favorable, much higher than is generally considered acceptable in practice.  We set $\gamma=0.1$.  We do not present the power curve for the Canonne test until it stops summarily rejecting all inputs, which for these parameters happens at $n=359$. We consider two possible effects, one where the true mean differs by 0.1 standard deviations in all coordinates, and one where it differs by 0.5 in only a single coordinate.  The results can be seen in Figure \ref{fig:Canonne}.

In the first case, with the larger effect size, we reach 80\% power at $n=65$, while the Canonne test isn't even valid until $n=359$ and doesn't reach 80\% power until $n=6500$.  In the second case, with a smaller effect size, the difference is smaller though still substantial.  ToT requires only 20\% as much data to reach 80\% power --- 190 data points compared to 850.  (At 99\% power, the gap is smaller, with ToT needing 87\% of the data needed by Canonne.) Additional comparisons are provided in Appendix \ref{sec:add_figs}.

Of course, the Canonne test \textit{does} have better asymptotic performance, so there is some sufficiently small $\varepsilon$ and effect size (and sufficiently large $\delta$) such that it becomes the higher-power test.  However, for a wide variety of practical situations, those superior asymptotics have not yet come into play, and ToT is the better choice.

\subsection{One-way Analysis of Variance} \label{ex:anova}

As our second example we consider a one-way ANOVA, which examines whether groups of data have the same mean. Formally, each of $g$ groups has a mean $\mu_g$.  Data within each group is drawn from $\mathcal{N}(\mu_i, \sigma^2)$ for a fixed, unknown $\sigma$.  Under $H_0$ all groups have equal mean, while in $H_A$ some means differ.  The classical test for this setting uses the $F$ statistic, which follows a known distribution under $H_0$.  (For a more thorough introduction, see \cite{rice2007mathstats}.) For this analysis, we focus on the case of equal-sized groups. We call the ratio of the between-group variance and the within-group variance $\eta = \textsf{Var}(\mu_1, \ldots, \mu_g)/\sigma^2$ the effect size. The power of an ANOVA in this setting has a known solution available in most statistical software.

This setting is the subject of a significant line of work.  \cite{campbell2018differentially} give the first private test, which was later improved upon by \cite{swanberg2019improved} and then \cite{couch2019differentially}.  To the best of our knowledge, the private nonparametric test of Couch et al.~is the most powerful private test available in this setting and thus will serve as a benchmark for the performance of the test of tests.

\begin{figure}[t]
    \centering
    \includegraphics[width=\columnwidth]{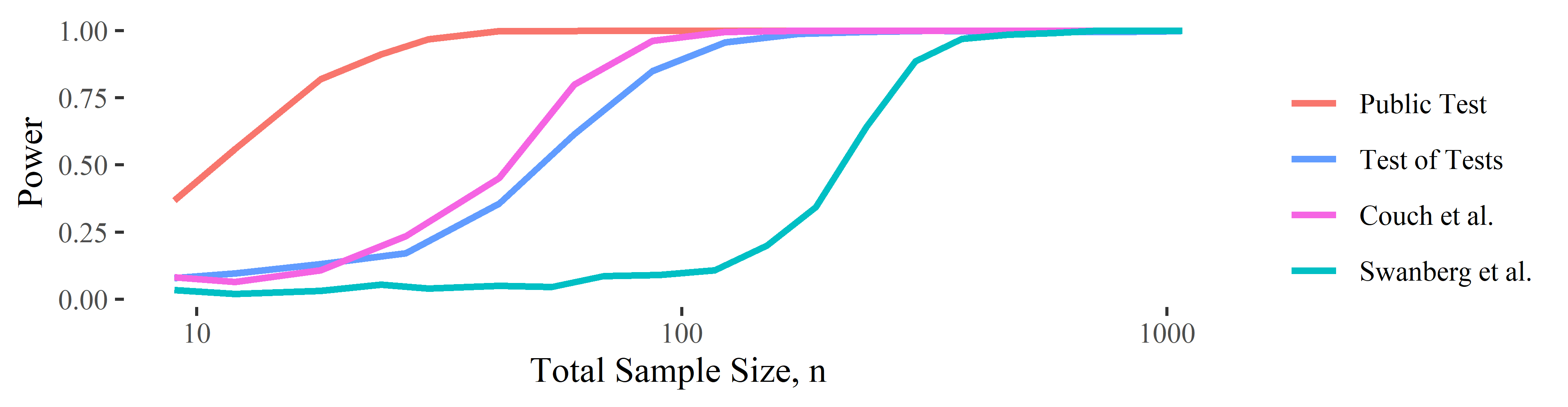}
    \caption{Power comparison between \cite{couch2019differentially}, \cite{swanberg2019improved}, and the test of tests for various sample sizes $n$. The effect size is $\eta = 1$, $\varepsilon = 1$, and number of groups $g = 3$. 
    All groups are of equal size and $\alpha = 0.05$. All non-public power curves are estimated via simulation. We optimize $m,\alpha_0$ for the test of tests as discussed in Section \ref{fwk:optim} with target $\rho = 0.9$.}\label{fig:anova}
\end{figure}
For comparison to the test of Couch et al., we choose the setting in Figure 3 of \cite{couch2019differentially} which examines privacy level $\epsilon = 1$ and effect size $\eta = 1$. As shown in the top panel of Figure \ref{fig:anova}, the test of tests is slightly worse, but the difference is small.  (Couch et al.~require 22\% less data to reach 80\% power.)
It performs much better than the test of Swanberg et al. Importantly, unlike the tests tailored to this setting, the test of tests does not require the estimation of a reference distribution via simulation. Thus, the test of tests is faster to run (and p-values are arguably more accurate).

Varying the setting shows that these two tests are incomparable.  We have included additional examples in Appendix \ref{sec:add_figs}.  With a smaller effect size, the gap between ToT and prior work increases, while a large effect size and/or smaller $\varepsilon$ actually results in ToT becoming the state of the art most powerful test, though by a small margin.  Regardless of the specifics of the comparison, we find it exciting that our general framework is at all comparable to a highly-refined test carefully developed for a specific situation.

\subsection*{Acknowledgments}

We would like to thank Andr\'es Barrientos for sharing code that implements their parallel work and Canyon Foot for writing some of the original functions that we still use. All authors were supported by the National Science Foundation under Grant No. SaTC-1817245. Kaiyan Shi acknowledges additional support from the U.S. Army Research Office under Grant No. W911NF-20-1-0015 and Zeki Kazan from NSF Grant No. SES-2217456.

\bibliography{bibliography.bib}
\bibliographystyle{plainnat}

\newpage
\appendix
\onecolumn
\section{Proofs}\label{sec:proofs}

Here we include the proofs that were excluded from the main body.

\subsection{Proofs for Section \ref{sec:framework}}

We now prove that the sample complexity of our test is an $\mathcal{O}(1/\varepsilon)$ factor more than that of the public test.

\begin{reptheorem}{thm:sc}
The sample complexity required for our test to achieve $\rho$ power is $$n = \mathcal{O} \left(\frac{1}{\varepsilon}\right)$$
\end{reptheorem}

\begin{proof}
Consider an alternative test with two changes: we add Laplace noise instead of Tulap noise and we use the proportion below the threshold as our test statistic, rather than the count. In line 6 of Algorithm \ref{alg:general}, we use the alternative test statistic
$$\Tilde{z} = \frac{a+\textup{Lap}\left(\frac{1}{\varepsilon}\right)}{m} = \frac{a}{m}+\frac{\textup{Lap}\left(\frac{1}{\varepsilon}\right)}{m}.$$
This output is guaranteed to be $\varepsilon$-differentially private by Theorems \ref{thm:pp} and \ref{thm:lm}. It has sensitivity $1$ since changing a row in the dataset can only change the p-value in one group and therefore can change $a$ by at most $1$. 

Let $L = \frac{\textup{Lap}\left(\frac{1}{\varepsilon}\right)}{m}$. Then there exists some $\varepsilon = \varepsilon^*$, some number of subgroups $m$, and some subgroup size $n$ such that $L = \frac{\textup{Lap}\left(\frac{1}{\varepsilon^*}\right)}{m}$ is small enough that power $\rho$ can be achieved with sample size $mn$, where $n$ is the number of datapoints in each group and $m$ is the number of groups.

Now let $\varepsilon' = \frac{\varepsilon^*}{k}$ and $m' = km$. By the scaling property of Laplace distribution, we have 
\begin{align*}
     \frac{\textup{Lap}\left(\frac{1}{\varepsilon'}\right)}{m'} = \frac{\textup{Lap}\left(\frac{1}{\varepsilon^*/k}\right)}{km} =  \frac{k\textup{Lap}\left(\frac{1}{\varepsilon^*}\right)}{km} =\frac{\textup{Lap}\left(\frac{1}{\varepsilon^*}\right)}{m}  = L.
\end{align*}
This means that if the number of datapoints in each group, $n$, is unchanged, then for $\varepsilon' =  \frac{\varepsilon^*}{k}$ the noise added is still $L$ if there are $k$ times as many groups.

We then need to consider how the additional groups affect the term $\frac{a}{m}$ (this term is independent of $\varepsilon$). Note that $\textup{E}[\frac{a}{m}] = \theta$, a constant, and so changing the number of groups will have no effect on this expectation. But $\textup{Var}[\frac{a}{m}] = \frac{\theta(1-\theta)}{m}$, so increasing the number of groups by a factor of $k$ will decrease the variance by a factor of $k$. Thus, the distribution of $\tilde{z}$ will have unchanged mean, and will still be distributed according to a binomial distribution, but it will now have lower variance. This means the power of the test must necessarily increase.

This analysis shows that if power $\rho$ can be achieved with $\varepsilon$ at sample size $mn$, then it can also be achieved with $\varepsilon / k $ at sample size smaller than or equal to $kmn$. In other words, the sample complexity is inversely related to $\varepsilon$. This shows that this alternate test has sample complexity $\mathcal{O} \left(\frac{1}{\varepsilon}\right).$ Our test, which uses the uniformly most powerful $\varepsilon$-differentially private binomial test instead of simple Laplace noise, must also have sample complexity $\mathcal{O} \left(\frac{1}{\varepsilon}\right).$
\end{proof}

Now we compute the exact (rather than asymptotic) power of the test.

Let $f_X(x)$ and $F_X(x)$ refer to the probability density function (PDF) and cumulative density function (CDF), respectively, of a random variable $X$. We now establish two lemmas about the CDFs of relevant variables.

\begin{lemma}\label{lma:AZ}
Let $\theta$ be the power of the public test $\tau$ in each of the $m$ subsamples with significance level $\alpha_0$. Let $A \sim \textup{Binomial}(m,\theta)$ and 
let $Z|A \sim \textup{Tulap}(A, e^{-\varepsilon})$. Then the cumulative distribution function of $Z$ is 
$$ F_Z(z) = \sum_{i = 0}^M F_{Z|A}(z | i)~f_A(i).$$
\end{lemma}

\begin{proof}
Let $f(a, z)$ be the joint probability density function of A and Z. Then the CDF of $Z$ is
\begin{align*}
    F_Z(z) &= \int_{-\infty}^z f_Z(t)~ dt \\
    &= \int_{-\infty}^z \sum_{i = 0}^M f(i, t) ~ dt \\
    &= \sum_{i = 0}^M \int_{-\infty}^z f_{Z|A}(t | i)~f_A(i) ~ dt \\
    &= \sum_{i = 0}^M F_{Z|A}(z | i)~f_A(i)
\end{align*}
\end{proof}

\begin{lemma} \label{lma:BN}
Let $B \sim \textup{Binomial}(m, \alpha_0)$ and $N \sim \textup{Tulap}(0, e^{-\varepsilon})$. Then the cumulative distribution function of $B+N$ is
$$ F_{B + N}(t) = \sum_{i = 0}^M f_B(i) F_N(i - t). $$
\end{lemma}

\begin{proof}
By convolution, the CDF of $B+N$ is
\begin{align*}
    F_{B + N}(t) &= P(B + N \leq t) \\
    &= \sum_{i = 0}^M \int_{t - i}^\infty f_B(i) f_N(x) ~ dx \\
    &= \sum_{i = 0}^M f_B(i) (1 - F_N(t - i)) \\ 
    &= \sum_{i = 0}^M f_B(i) F_N(i - t)
\end{align*}
\end{proof}

\begin{reptheorem}{thm:power}
Let $\theta$ be the power of the public test $\tau$ in each of the $m$ subsamples with significance level $\alpha_0$. Let $A \sim \textup{Binomial}(m,\theta)$, $Z|A \sim \textup{Tulap}(A, e^{-\varepsilon})$, $B \sim \textup{Binomial}(m, \alpha_0)$, and $N \sim \textup{Tulap}(0, e^{-\varepsilon})$.
Then the power of our test is
$$\mathcal{P}(\varepsilon, \alpha, m, \alpha_0, \theta) = (1 - F_{Z}(F^{-1}_{B+N}(1-\alpha))).$$
\end{reptheorem}

\begin{proof}
Let $W_i$ be a random variable which outputs $1$ if $p_{i} < \alpha_0$ and $0$ otherwise. It is thus distributed $W_i \sim \textup{Bernoulli} (\theta)$. Then $A = \sum_{i = 1}^m W_i \sim \textup{Binomial}(m, \theta)$ 
is the number of p-values less than $\alpha_0$ and $Z \mid A \sim \textup{Tulap}(A, e^{-\varepsilon})$ is the differentially private estimate of $A$.

$B \sim \textup{Binomial}\left(M, \alpha_0 \right)$ is the number of p-values less than $\alpha_0$ under the null hypothesis and $N \sim \textup{Tulap}(0, e^{-\varepsilon}, 0)$ is the required amount of Tulap noise to maintain $\varepsilon$-differential privacy. For any valid hypothesis test, under the null hypothesis, $\theta = P(p_i < \alpha_0) \leq \alpha_0$. Testing the hypothesis of interest is thus equivalent to testing
\begin{align*}
    H_0 : \theta \leq \alpha_0 \quad \textup{and} \quad H_A : \theta > \alpha_0.
\end{align*}
The p-value for this test is
$$ p(Z) = P(B + N \geq Z \mid Z). $$
The power of this test is then
\begin{align}
    P(p(Z) \leq \alpha) &= P(1 - F_{B+N}(Z) \leq \alpha) \nonumber \\
    &=P(Z \geq F^{-1}_{B+N}(1- \alpha)) \nonumber\\
    &= (1 - F_{Z}(F^{-1}_{B+N}(1-\alpha))). \nonumber
\end{align}
By Lemmas \ref{lma:AZ} and \ref{lma:BN}, $F_Z(z) = \sum_{i = 0}^M F_{Z|A}(z | i)~f_A(i)$ and $F_{B + N}(t) = \sum_{i = 0}^M f_B(i) F_N(i - t)$. This completes the proof.
\end{proof}

Fixing the public test, we then get this corollary, from which we can calculate bounds on the cost of privacy, thought of as the increase in the amount of data needed compared to the non-private test.

\begin{repcorollary}{cor:general_power}
Suppose that a public hypothesis test $\tau$ requires at most $n$ data points to achieve $\theta$ power at a significance level $\alpha_0$ for any choice of the data. Then, in order for the private test with privacy parameter $\varepsilon$ to achieve $\rho$ power at a significance level $\alpha$, the necessary number of data points is bounded above by $n\tilde{m}$, where $\tilde{m}$ is the smallest $m$ such that 
$$ \rho \leq \mathcal{P}(\varepsilon, \alpha, m, \alpha_0, \theta)$$
\end{repcorollary}

\begin{proof}
Consider a database partitioned into $\tilde{m}$ subsets, each of size $n$. When running the public hypothesis test on each subset, the true probability that the p-value is over the threshold $\alpha_0$ is some $\theta^* \geq \theta$. Consider $\mathcal{P}(\varepsilon, \alpha, \tilde{m}, \alpha_0, \theta^*)$. Since the distribution of $Z$ under the alternative hypothesis will shift further away from the null distribution (its center is now $\tilde{m}\theta^* \geq \tilde{m}\theta \geq \tilde{m}\alpha_0$), it follows that
$$ \mathcal{P}(\varepsilon, \alpha, \tilde{m}, \alpha_0, \theta^*) \geq \mathcal{P}(\varepsilon, \alpha, \tilde{m}, \alpha_0, \theta). $$

Now let $m'$ be the smallest $m$ such that
$$ \rho \leq \mathcal{P}(\varepsilon, \alpha, m, \alpha_0, \theta^*).$$
Since $\mathcal{P}$ is strictly increasing as a function of $m$, it follows that $m' \leq \tilde{m}$.

Now consider the true number of datapoints required for the private test to achieve $\rho$ power. I.e., the minimum number of datapoints the test can achieve full power over all choices of $\alpha_0$ and $m$. Formally, we define
\begin{align*}
    m^* = \argmin_{m} \{ \rho \leq \mathcal{P}(\varepsilon, \alpha, m, \alpha_0, \theta^*) \mid m \in \mathbb{N}, \alpha_0 \in [0,1] \}. 
\end{align*}

Then since the test can achieve $\rho$ power with $m'n$ datapoints and $m^*n$ is the minimum number of datapoints required to achieve $\rho$ power, it follows that $m'n \geq m^*n$. Combining this with the early inequality gives $\tilde{m}n \geq m^*n$.
\end{proof}

\subsection{Proofs for Section \ref{sec:PB}}

The results of \cite{awan2018differentially} can be used to show that ToT has higher power than the PB framework.

Here we compute an analytic expression for the power of the test proposed by \cite{pena2022differentially}.

\begin{theorem} \label{thm:PB_power}
Let $\theta$ be the power of the public test $\tau$ with significance level $\alpha_0$. Let $f_{i,p,m}$ be the probability mass function of a Poisson-binomial distribution with a success probability vector of $p$ repeated $i$ times and $1-p$ repeated $m-i$ times. Then the power of the PB test is
$$\mathcal{P}_{PB}(\varepsilon, \alpha, m, p, \alpha_0, \theta) = \sum_{i=0}^m \sum_{j = \frac{m+1}{2}}^m f_{i,p,m}(j) \binom{m}{i} \theta^i(1-\theta)^{m-i}.$$
\end{theorem}

\begin{proof}
    Let $W$ be the number of sub-samples in which the public test is rejected and let $\tilde{W}$ be the number in which the sub-test is rejected after the randomized response mechanism is applied. We begin by considering the probability $H_0$ is rejected conditional on $W = i$, which occurs if and only if $\tilde{W} > \frac{m-1}{2}$. The probability that $\tilde{W} = j$ is given by a Poisson-Binomial distribution with the vector of success probabilities
    \begin{equation*}
        (\underbrace{p, \ldots, p}_{i \textup{ times}}, \underbrace{1-p, \ldots, 1-p}_{m-i \textup{ times}}).
    \end{equation*}
    We let $f_{i,p,m}$ denote this distribution.

    Now consider the overall test. The probability that $W = i$ is given by a binomial distribution with size $m$ and probability $\theta$. Applying the Law of Total Probability then gives,
    \begin{align*}
        Pr\left(\tilde{W} > \frac{m-1}{2}\right) &= \sum_{i=0}^m Pr\left(\tilde{W} > \frac{m-1}{2} \mid W = i \right) \, Pr\left(W = i \right) \\
        &= \sum_{i=0}^m \sum_{j=\frac{m+1}{2}}^m Pr\left(\tilde{W} = j \mid W = i \right) \, \binom{m}{i} \theta^i(1-\theta)^{m-i} \\
        &= \sum_{i=0}^m \sum_{j=\frac{m+1}{2}}^m f_{i,p,m}(j) \binom{m}{i} \theta^i(1-\theta)^{m-i}.
    \end{align*}
    Note that the PB optimization will only select odd $m$, ensuring that $\frac{m+1}{2}$ is an integer. This completes the proof.
\end{proof}

\subsection{Proofs for Section \ref{sec:tailored}}

Here we compute an analytic expression for the power of the public test for deviation of $d$-dimensional Gaussian from a given mean.

\begin{reptheorem}{thm:pow_norm}
Let $F_0$ be the CDF of $\chi^2(df = d)$ and $F_A$ be the CDF of $\chi^2 \left(df = d, \lambda = n\sum_{j = 1}^d \mu_j^2 \right)$, where $n$ is the sample size and $\mu_j$ is the $j^{th}$ entry of $\boldsymbol{\mu}$. Then the power of the public test with significance level $\alpha$ is
$1 - F_{A}(F_0^{-1}(1-\alpha)$.
\end{reptheorem}

\begin{proof}
First, note that $\bar{X}_j \sim \mathcal{N}\left(\mu_j, \sigma = \frac{1}{\sqrt{n}}\right)$. Thus, $\sqrt{n}\bar{X}_j \sim \mathcal{N}\left(\sqrt{n}\mu_j, \sigma = 1 \right)$. It follows that, for the test statistic,
$$Z = n\sum_{j = 1}^d \bar{X}_j^2 = \sum_{j = 1}^d (\sqrt{n}\bar{X}_j)^2 \sim \chi^2(df = d, \lambda).$$
Under the null hypothesis, $\lambda = 0$. But if $\boldsymbol{\mu} \neq \mathbf{0}$, 
$$ \lambda = \sum_{j = 1}^d(\sqrt{n}\mu_j)^2 = n\sum_{j = 1}^d \mu_j^2. $$
Let $Z$ be the observed test statistic and $P$ be the corresponding p-value. For significance level $\alpha$, the power of the test is thus
\begin{align*}
    Pr(P \leq \alpha) &= Pr(F_0(Z) \geq 1- \alpha) \\
    &= Pr(Z \geq F^{-1}_0(1-\alpha) \\
    &= 1 - F_{A}(F_0^{-1}(1-\alpha)).
\end{align*}

\end{proof}

Here we compute a lower bound on the Type 1 error of the Canonne et al.~test.

\begin{theorem} \label{thm:canonne_TypeI}
    Let $L \sim \textup{Laplace}(b)$, where 
    \begin{align*}
        b &= \left(5 \Delta_\delta^G+\frac{432 d}{\varepsilon} \ln \frac{n d}{\delta} \sqrt{\ln \frac{n}{\delta} \cdot \ln \frac{5}{4 \delta}}\right) / \varepsilon 
    \end{align*}
    and $\Delta_\delta^G$ is as defined in Algorithm 4 of \cite{canonne2019private}. Let $Z \sim \mathcal{N}(0, \sigma = n\sqrt{2d})$. Then the Type I Error of the Canonne et al.~algorithm is bounded below by
    \begin{equation*}
        1 - F_{Z + L}\left(\frac{n^2\gamma^2}{324} \right).
    \end{equation*}
\end{theorem}

\begin{proof}
    We begin with Stage 2 of Canonne et al.'s Algorithm 4. In Stage 1, any condition that fails results in a rejection of the null hypothesis, which implies that the Type I Error resulting from the final steps is in a lower bound on the overall Type I Error. In Stage 2, for each row $j \in \{1, \ldots, n\}$, the algorithm will either draw $\hat{X}^{(j)}$ from $\mathcal{N}_d(\mathbf{0}, \mathbb{I}_d)$ or set $\hat{X}^{(j)} = X^{(j)}$, the original row. But under the null hypothesis, $X^{(j)} \sim \mathcal{N}_d(\mathbf{0}, \mathbb{I}_d)$, so either way $\hat{X}^{(j)} \sim \mathcal{N}_d(\mathbf{0}, \mathbb{I}_d)$.

    Under the null hypothesis, then, $T(\hat{X}) = T(X)$. As a consequence of the author's Theorem B.2, under the null hypothesis $T(\hat{X}) \sim \mathcal{N}(0, \sigma = n\sqrt{2d})$. In stage 3, the algorithm then adds noise from $L \sim \textrm{Laplace}(b)$ to $T(X)$, where 
    \begin{align*}
        b &= \left(5 \Delta_\delta^G+\frac{432 d}{\varepsilon} \ln \frac{n d}{\delta} \sqrt{\ln \frac{n}{\delta} \cdot \ln \frac{5}{4 \delta}}\right) / \varepsilon  \\
        \Delta_\delta^G &= 144\Bigg(d \ln \frac{d}{\delta}+\frac{d}{n \varepsilon^2} \ln ^2 \frac{1}{\delta} +\sqrt{n d} \sqrt{\ln \frac{d}{\delta} \cdot \ln \frac{n}{\delta}} +\frac{\sqrt{d}}{\varepsilon} \ln \frac{1}{\delta} \sqrt{\ln \frac{n}{\delta}}\Bigg) \ln \frac{n d}{\delta}.
    \end{align*}
    The null hypothesis is then rejected if and only if $T(\hat{X}) + L > \frac{n^2\gamma^2}{324}$. Letting $Z \sim \mathcal{N}(0, \sigma = n\sqrt{2d})$, the Type I Error of the test is then bounded below by
    \begin{align*}
        P\left(T(\hat{X}) + L > \frac{n^2\gamma^2}{324} \mid H_0\right) &= 1 - P\left(Z + L \leq \frac{n^2\gamma^2}{324}\right) \\
        &= 1 - F_{Z + L}\left(\frac{n^2\gamma^2}{324} \right).
    \end{align*}
\end{proof}

\section{Additional Figures}\label{sec:add_figs}

Here we include additional figures.

\FloatBarrier
\subsection{Pe{\~n}a-Barrientos Framework}

Figure \ref{fig:PB} presents a comparison of the theoretical power of the binomial tests proposed by \cite{awan2018differentially} and \cite{pena2022differentially}. Figures \ref{fig:App_PB_z} and \ref{fig:App_PB_ANOVA} present comparisons of the \cite{pena2022differentially} framework and the test of tests in additional settings.

\begin{figure}[h]
    \centering
    \includegraphics[width=\columnwidth]{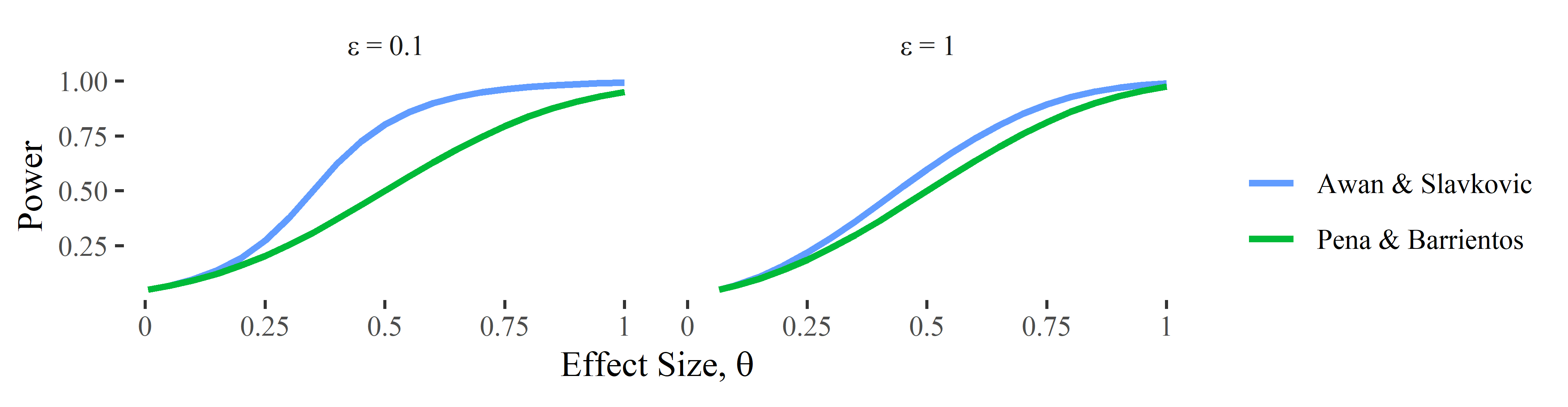}
    \caption{Let $m$ and $\alpha_0$ be the parameters selected for the PB test using the ``low-power" setting recommendations for $\alpha = 0.05$ and a given $\varepsilon$. The plot presents the power of a test of the hypotheses $H_0: \theta \leq \alpha_0$ and $H_A : \theta > \alpha_0$ for data $\data \sim \textup{Binom}(m, \theta)$ as a function of $\theta$ for the binomial tests proposed in \cite{awan2018differentially} and \cite{pena2022differentially}. The left panel presents $\varepsilon = 0.1$, and the right panel presents $\varepsilon = 1$. }
    \label{fig:PB}
\end{figure}

\begin{figure}[h]
    \centering
    \includegraphics[width=\columnwidth]{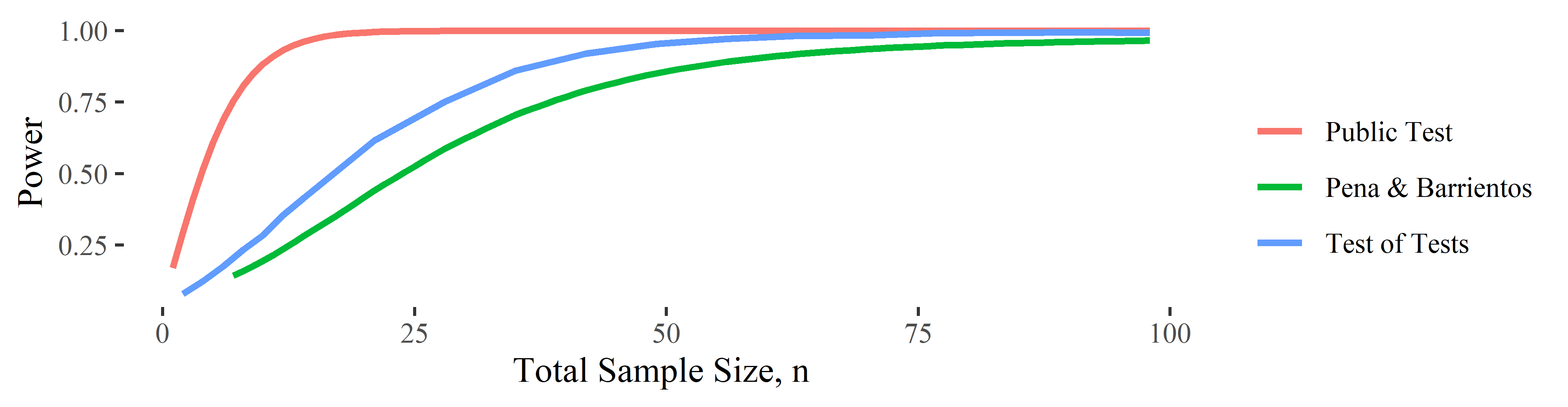}
    \caption{Power comparison between \cite{pena2022differentially} and the test of tests for various choices of database size $n$ with $\alpha = 0.05$ and a z-test. The effect size is $\mu/\sigma = 1$ and privacy parameter is $\varepsilon = 1$. We optimize $m$ and $\alpha_0$ for the test of tests at each $n$ as discussed in Section \ref{fwk:optim} with target power $\rho = 0.9$.}
    \label{fig:App_PB_z}
\end{figure}

\begin{figure}[h]
    \centering
    \includegraphics[width=\columnwidth]{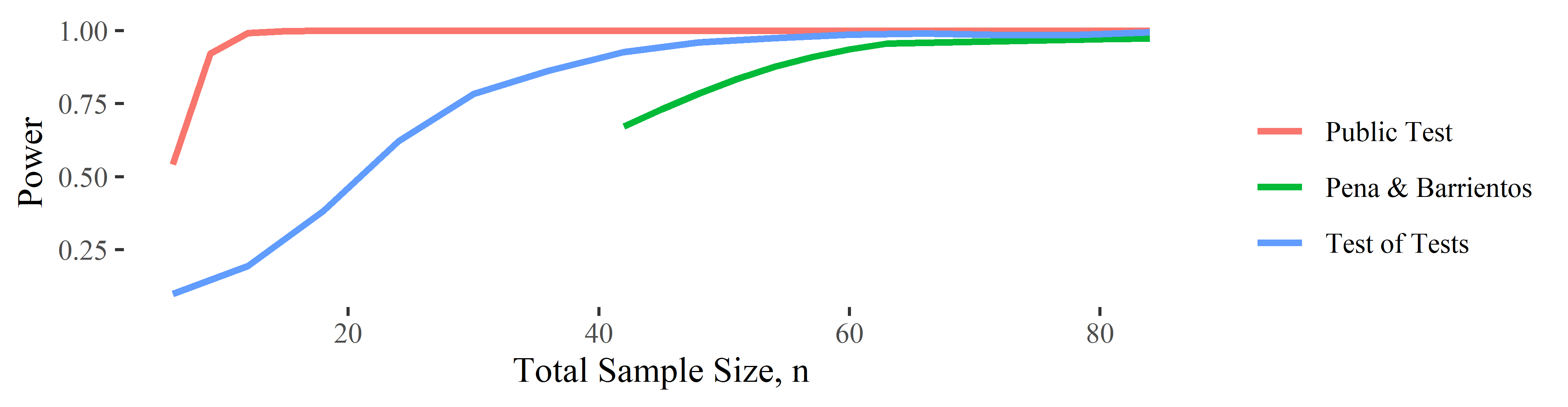}
    \caption{Power comparison between \cite{pena2022differentially} and the test of tests for various choices of database size $n$ with $\alpha = 0.05$ and a one-way ANOVA with non-private groups. The effect size is $\eta = 4$ and privacy parameter is $\varepsilon = 1$. We optimize $m$ and $\alpha_0$ for the test of tests at each $n$ as discussed in Section \ref{fwk:optim} with target power $\rho = 0.9$.}
    \label{fig:App_PB_ANOVA}
\end{figure}

\FloatBarrier
\subsection{Multivariate Normal Data}

Figure \ref{fig:Canonne_TypeI} compares the empirical Type I Error of the two tests in Figure \ref{fig:Canonne}. Figures \ref{fig:Canonne_1} to \ref{fig:Canonne_4} provide additional power comparisons between \cite{canonne2019private} and ToT with various dimensions $d$, effect sizes $\boldsymbol{\mu}$, and privacy parameters $\varepsilon$.

\begin{figure}[h]
    \centering
    \includegraphics[width=\columnwidth]{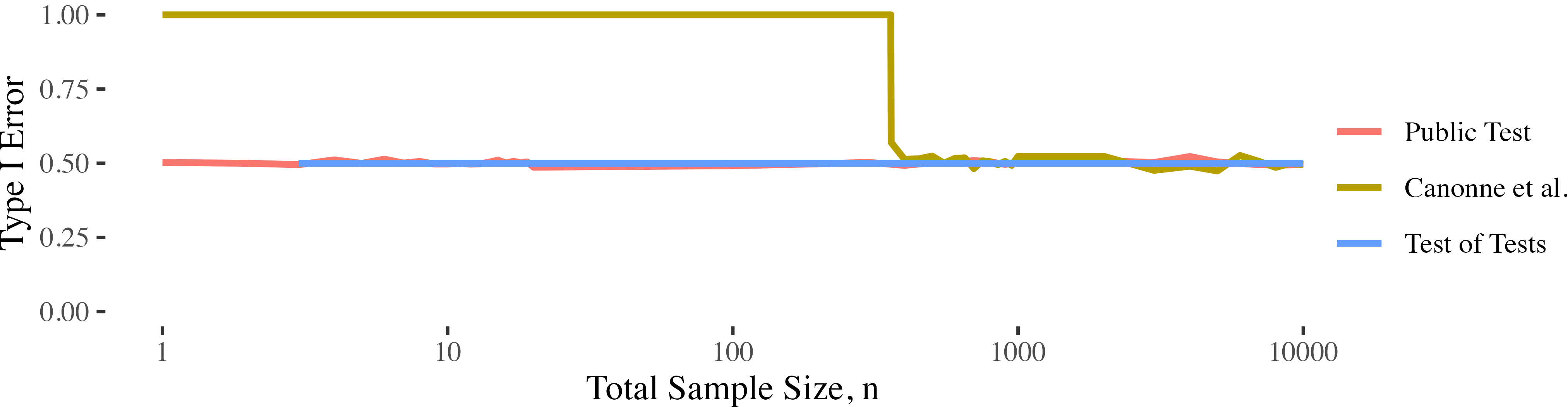}
    \caption{Type I error comparison between \cite{canonne2019private} and the tests of tests for various sample sizes $n$. The dimension is $d=100$, $\varepsilon = 1$, and $\alpha$ for the test of tests and public test is set to be $0.5$. We optimize $m$ and $\alpha_0$ for the test of tests at each $n$ as discussed in Section \ref{fwk:optim} with target $\rho = 0.9$.}
    \label{fig:Canonne_TypeI}
\end{figure}

\begin{figure}[h]
    \centering \includegraphics[width=\columnwidth]{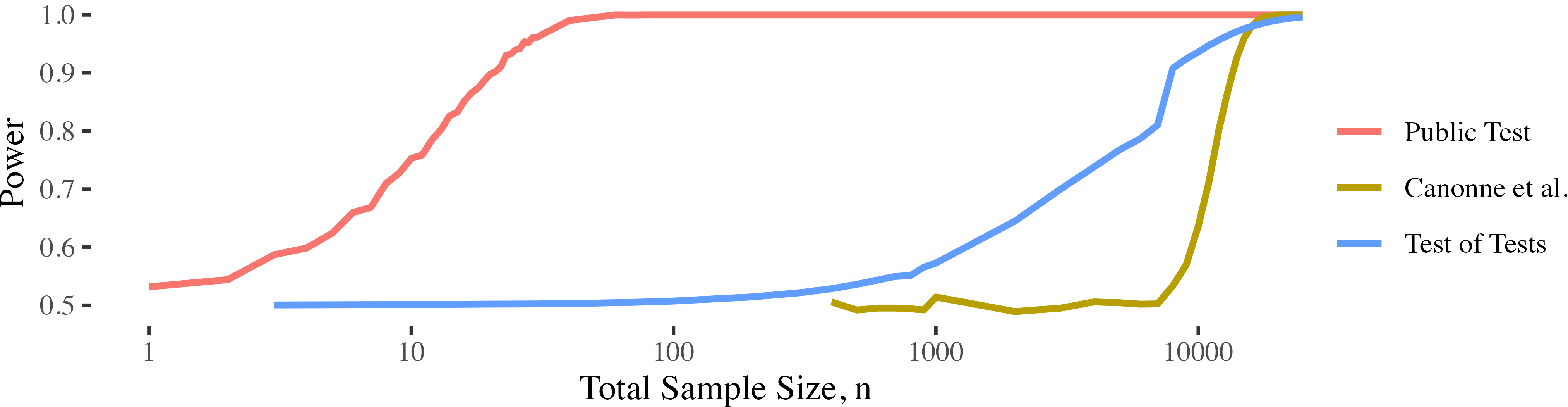}
    \caption{Power comparison between \cite{canonne2019private} and the tests of tests for various sample sizes $n$. The true mean is $\mu_1 = 0.1$ and $\mu_i = 0$ for all $i \neq 1$. The dimension is $d=100$, $\varepsilon = 1$, and $\alpha$ for the test of tests and public test is set to match the Type I Error of the \cite{canonne2019private} test. We optimize $m$ and $\alpha_0$ for the test of tests at each $n$ as discussed in Section \ref{fwk:optim} with target $\rho = 0.9$.}
    \label{fig:Canonne_1}
\end{figure}

\begin{figure}[h]
    \centering \includegraphics[width=\columnwidth]{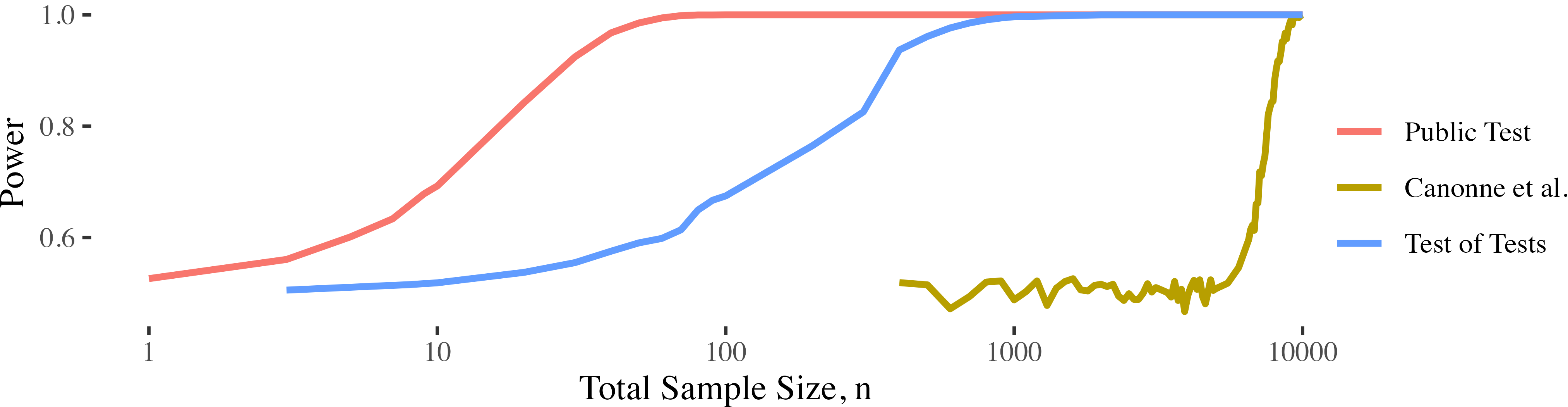}
    \caption{Power comparison between \cite{canonne2019private} and the tests of tests for various sample sizes $n$. The true mean is $\mu_i = 0.1$ for $i <=20$ and $\mu_i = 0$ otherwise. The dimension is $d=60$, $\varepsilon = 1$, and $\alpha$ for the test of tests and public test is set to match the Type I Error of the \cite{canonne2019private} test. We optimize $m$ and $\alpha_0$ for the test of tests at each $n$ as discussed in Section \ref{fwk:optim} with target $\rho = 0.9$.}
    \label{fig:Canonne_2}
\end{figure}

\begin{figure}[h]
    \centering
    \includegraphics[width=\columnwidth]{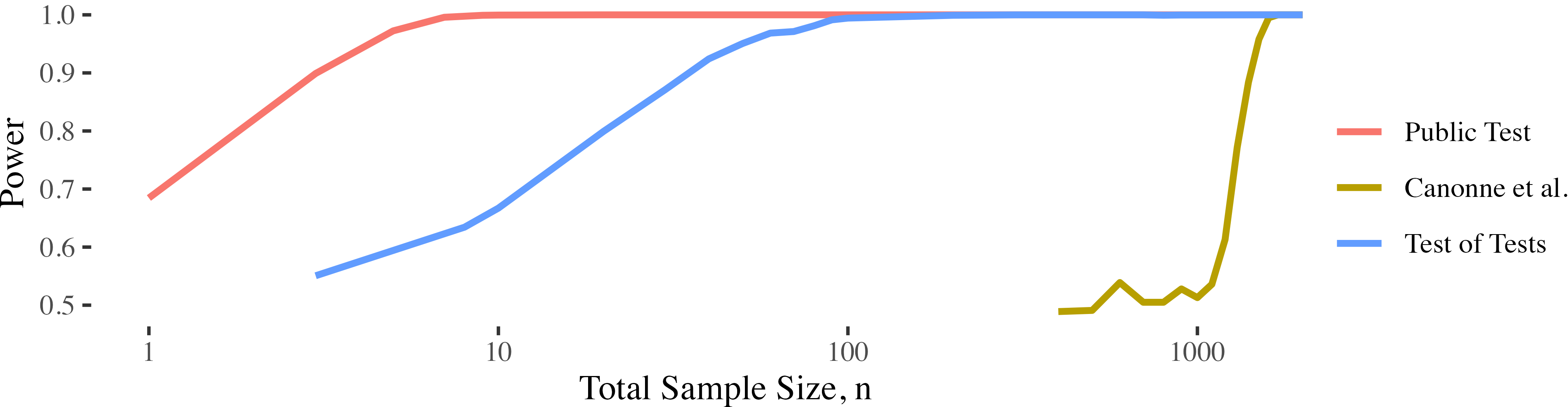}
    \caption{Power comparison between \cite{canonne2019private} and the tests of tests for various sample sizes $n$. The true mean is $\mu_i = 0.3$ for $i <=20$ and $\mu_i = 0$ otherwise. The dimension is $d=60$, $\varepsilon = 1$, and $\alpha$ for the test of tests and public test is set to match the Type I Error of the \cite{canonne2019private} test. We optimize $m$ and $\alpha_0$ for the test of tests at each $n$ as discussed in Section \ref{fwk:optim} with target $\rho = 0.9$ for $n \leq 50$ and $\rho=0.99$ for $n > 50$.}
    \label{fig:Canonne_3}
\end{figure}

\begin{figure}[h]
    \centering
    \includegraphics[width=\columnwidth]{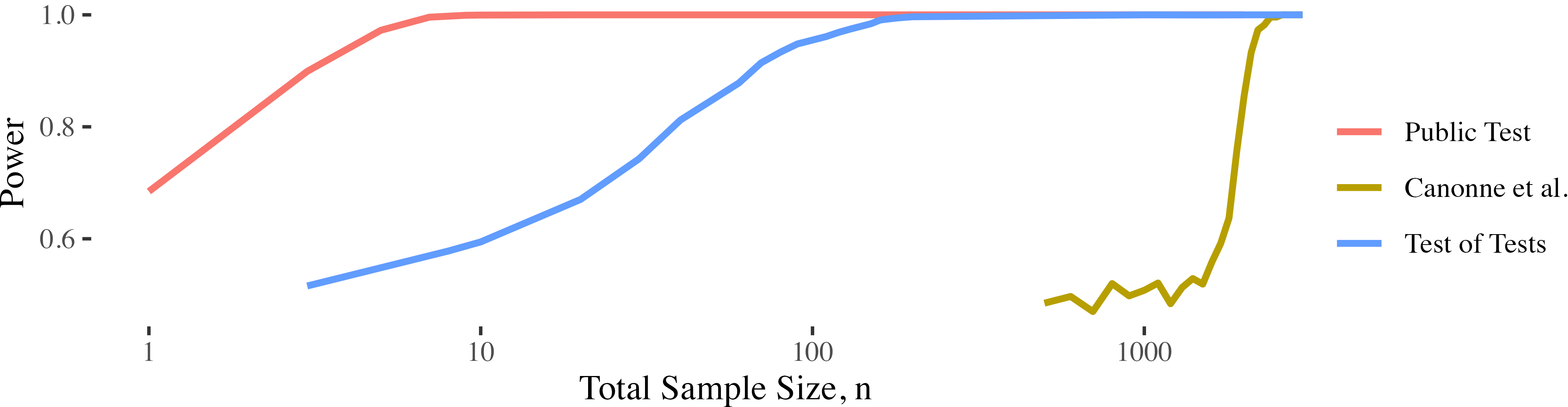}
    \caption{Power comparison between \cite{canonne2019private} and the tests of tests for various sample sizes $n$. The true mean is $\mu_i = 0.3$ for $i <=20$ and $\mu_i = 0$ otherwise. The dimension is $d=60$, $\varepsilon = 0.5$, and $\alpha$ for the test of tests and public test is set to match the Type I Error of the \cite{canonne2019private} test. We optimize $m$ and $\alpha_0$ for the test of tests at each $n$ as discussed in Section \ref{fwk:optim} with target $\rho = 0.9$ for $n \leq 100$ and $\rho=0.99$ for $n > 100$.}
    \label{fig:Canonne_4}
\end{figure}

\FloatBarrier
\subsection{One-way Analysis of Variance}

Figures \ref{fig:couch_group_2_eff_0.5} to \ref{fig:couch_group_3_eff_25} give more comparisons between \cite{couch2019differentially}. \cite{swanberg2019improved}, and ToT for various choices of parameters. Figures \ref{fig:couch_group_2_eff_0.5} and \ref{fig:couch_group_2_eff_25} are comparisons with $g=2$ groups with both effect sizes, while Figures \ref{fig:couch_group_3_eff_0.35} and \ref{fig:couch_group_3_eff_25} are comparisons with $g=3$ groups.

\begin{figure}[h]
    \centering
    \includegraphics[width=\columnwidth]{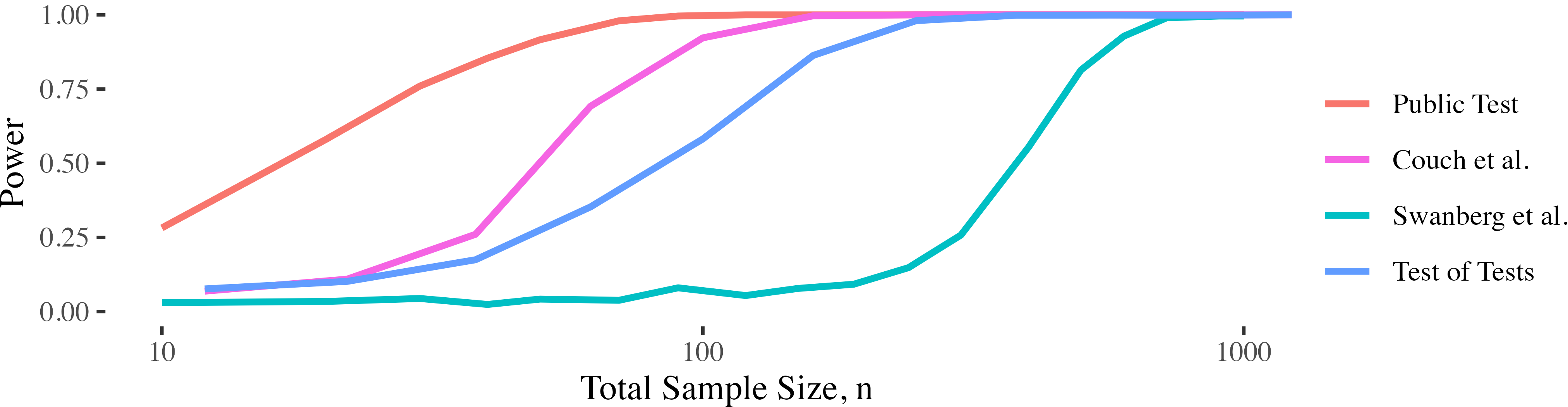}
    \caption{Power comparison between \cite{couch2019differentially}. \cite{swanberg2019improved}, and the test of tests for various choices of database size $n$. The effect size is $\eta = 0.5$, privacy parameter $\varepsilon = 1$, and number of groups $g = 2$. All groups are of equal size and $\alpha = 0.05$. All power curves (except the public test) are estimated via simulation. We optimize $m$ and $\alpha_0$ for the test of tests at each $n$ as discussed in Section \ref{fwk:optim} with target $\rho = 0.9$.}\label{fig:couch_group_2_eff_0.5}
\end{figure}

\begin{figure}[h]
    \centering
    \includegraphics[width=\columnwidth]{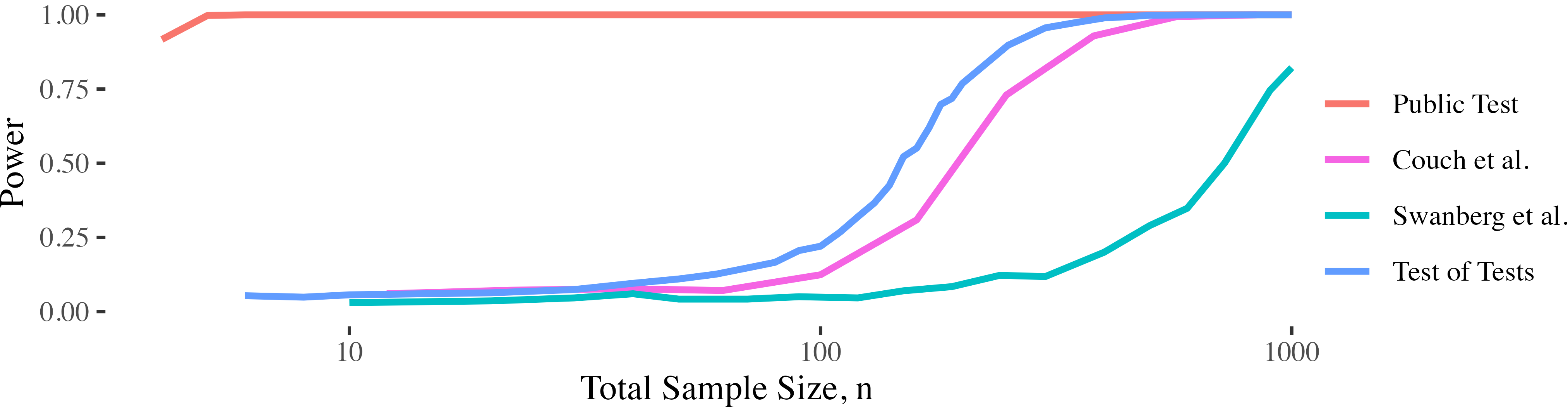}
    \caption{Power comparison between \cite{couch2019differentially}. \cite{swanberg2019improved}, and the test of tests for various choices of database size $n$. The effect size is $\eta = 25$, privacy parameter $\varepsilon = 0.1$, and number of groups $g = 2$. All groups are of equal size and $\alpha = 0.05$. All power curves (except the public test) are estimated via simulation. We optimize $m$ and $\alpha_0$ for the test of tests at each $n$ as discussed in Section \ref{fwk:optim} with target $\rho = 0.9$.}\label{fig:couch_group_2_eff_25}
\end{figure}

\begin{figure}[h]
    \centering
    \includegraphics[width=\columnwidth]{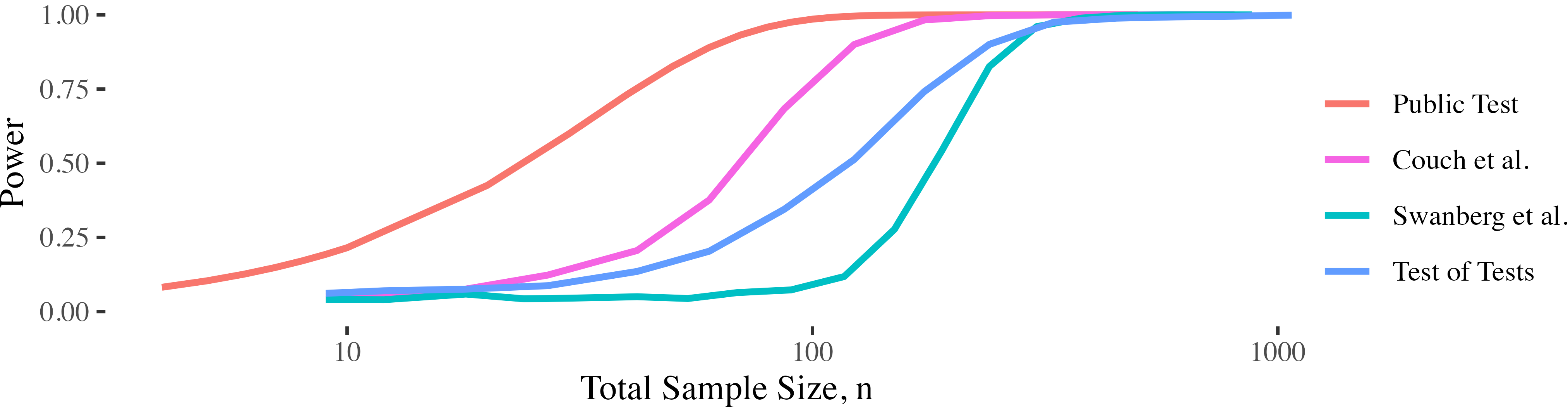}
    \caption{Power comparison between \cite{couch2019differentially}. \cite{swanberg2019improved}, and the test of tests for various choices of database size $n$. The effect size is $\eta = 0.35$, privacy parameter $\varepsilon = 1$, and number of groups $g = 3$. All groups are of equal size and $\alpha = 0.05$. All power curves (except the public test) are estimated via simulation. We optimize $m$ and $\alpha_0$ for the test of tests at each $n$ as discussed in Section \ref{fwk:optim} with target $\rho = 0.9$.}\label{fig:couch_group_3_eff_0.35}
\end{figure}

\begin{figure}[h]
    \centering
    \includegraphics[width=\columnwidth]{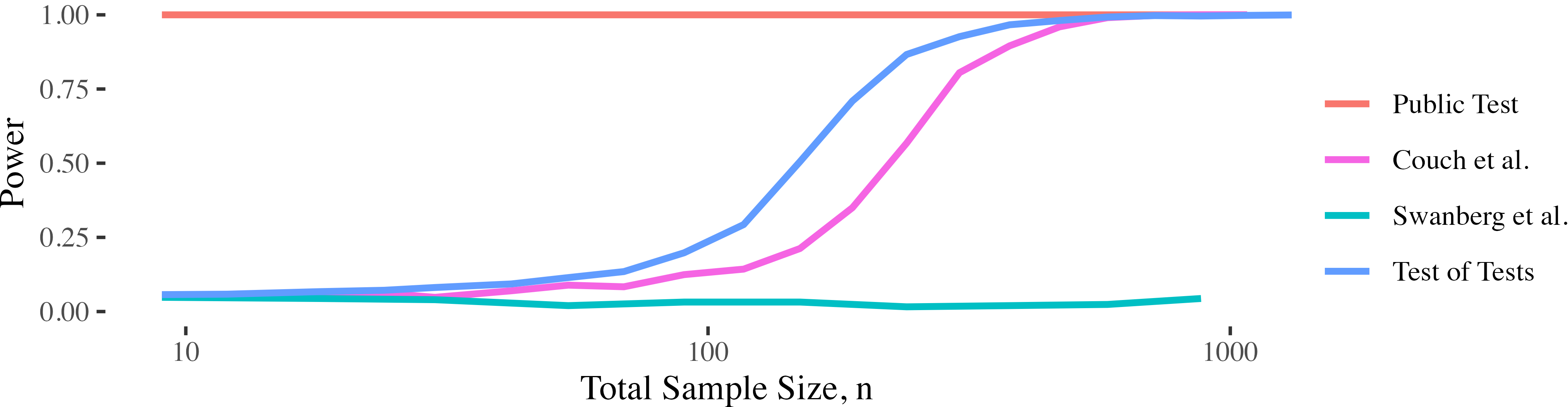}
    \caption{Power comparison between \cite{couch2019differentially}. \cite{swanberg2019improved}, and the test of tests for various choices of database size $n$. The effect size is $\eta = 25$, privacy parameter $\varepsilon = 0.1$, and number of groups $g = 3$. All groups are of equal size and $\alpha = 0.05$. All power curves (except the public test) are estimated via simulation. We optimize $m$ and $\alpha_0$ for the test of tests at each $n$ as discussed in Section \ref{fwk:optim} with target $\rho = 0.9$.}\label{fig:couch_group_3_eff_25}
\end{figure}

\end{document}